\documentclass[11pt,a4paper]{article}
\usepackage[margin=1in]{geometry}
\usepackage[T1]{fontenc}
\usepackage{lmodern}
\usepackage{microtype}
\usepackage{amsmath,amssymb,amsthm,mathtools}
\usepackage{needspace}
\usepackage{hyperref}
\hypersetup{colorlinks=true,linkcolor=black,citecolor=black,urlcolor=black,
  pdftitle={Almost every Boolean function has central threshold degree},
  pdfauthor={Martin Anthony}}
\newtheorem{theorem}{Theorem}
\newtheorem{lemma}[theorem]{Lemma}
\newcommand{\F}{\mathbb F}
\newcommand{\R}{\mathbb R}
\newcommand{\E}{\mathbb E}
\newcommand{\rank}{\operatorname{rank}}
\newcommand{\supp}{\operatorname{supp}}
\newcommand{\RM}{\operatorname{RM}}
\newcommand{\td}{\operatorname{td}}
\title{Almost every Boolean function has central threshold degree}
\author{Martin Anthony\thanks{Department of Mathematics,
London School of Economics and Political Science.\newline
Email: \href{mailto:m.anthony@lse.ac.uk}{\texttt{m.anthony@lse.ac.uk}}.}}
\date{}
\begin{document}
\maketitle
\vspace{-2\baselineskip}
\enlargethispage{\baselineskip}

\begin{abstract}
The polynomial threshold degree (or threshold degree) of a Boolean
function is the least degree of a real polynomial whose sign represents it.
We prove that almost every Boolean function has polynomial threshold
degree among the central one or two integers. More precisely, a uniformly
chosen Boolean function on $n$ variables has threshold degree either
$\lfloor n/2\rfloor$ or $\lceil n/2\rceil$, except with probability at most
$\exp(-c2^{n/2}/n)$, for an absolute constant $c>0$ and all sufficiently
large $n$.
The paper proves a 1991 conjecture of Aspnes, Beigel, Furst and Rudich.
It also proves a 1991 conjecture of Wang and Williams for even $n$ and
its main concentration assertion for odd $n$.
A parity example gives a lower bound on the exceptional probability
whose exponent differs from ours by a factor of order $n$.
The proof relates polynomial sign representation to recovery of
Reed--Muller codes from random erasures, and combines sharp-threshold
estimates with bounds on generalized Hamming weights and minimum distance.
\end{abstract}

\section{Introduction}\label{sec:introduction}

The polynomial threshold degree (or simply threshold degree) $\td(f)$
of a Boolean function $f:\{0,1\}^n\to\{-1,1\}$ is the least degree of a real polynomial
$P$ such that $f(x)P(x)>0$ at every vertex. Weak threshold degree is
defined similarly, but allows $f(x)P(x)\geq0$, provided $P$ is not
identically zero on the cube. By counting regions in a
hyperplane arrangement, Anthony~\cite{Anthony} and, independently, Alon
(reported by Saks~\cite{Saks}) established a lower bound of
$\lfloor n/2\rfloor$ for almost every Boolean function.

Earlier upper bounds for almost every Boolean function included
$19n/20$, due to Razborov and Rudich, and $0.89n$, obtained by Alon
using results of Gotsman. See O'Donnell and Servedio~\cite{OS} for an
account of these results. O'Donnell and Servedio proved the upper bound
$n/2+O(\sqrt{n\ln n})$, and also reported an independent proof of this
bound by Samorodnitsky, citing an unpublished manuscript from 2000 and
a personal communication from 2003. Together with the counting lower
bound, this determined the leading term of the typical threshold degree.

We prove the conjecture of Aspnes, Beigel, Furst and
Rudich~\cite{ABFR} that almost every Boolean function on an even number
$n$ of variables has both threshold degree and weak threshold degree
equal to $n/2$ (Section~\ref{sec:consequences}). We also prove the conjecture of Wang and
Williams~\cite[Section~6.4, Equation~16]{WW} for even $n$. For odd $n$,
we prove their main concentration assertion: almost every Boolean
function has threshold degree either $(n-1)/2$ or $(n+1)/2$. Their
further assertion that the proportions at these two degrees both tend
to $1/2$ remains open.

\begin{theorem}\label{thm:main}
There is an absolute constant $c>0$ such that, for all sufficiently
large $n$, a uniformly random Boolean function $f$ on $n$ variables satisfies
\[
 \Pr\left(\left\lfloor\frac n2\right\rfloor\leq\td(f)
       \leq\left\lceil\frac n2\right\rceil\right)
 \geq1-\exp\left(-c\frac{2^{n/2}}n\right).
\]
In particular, almost every Boolean function on an even number $n$
of variables has threshold degree exactly $n/2$.
\end{theorem}

Baldi and Vershynin~\cite{BV} used results of Abbe, Shpilka and
Wigderson~\cite{ASW} on Reed--Muller codes to establish linear
independence estimates in their enumeration of polynomial threshold
functions. We pursue this connection for the typical-degree problem.
For the upper bound, we use results on recovering binary codewords when
some of their coordinates have been erased. In Section~\ref{sec:erasure-recovery}, we show that the probability
of failure of sign representation is at most twice the probability that
erasures prevent unique recovery in a particular Reed--Muller code.
We then show that the probability of recovery failure is exponentially
small, using sharp-threshold estimates and bounds on the supports of codewords.

Passing from real polynomials to binary codes gives us more symmetry.
Affine changes of variables over $\F_2$ preserve Reed--Muller codes,
but need not preserve real polynomial degree. We use this additional
symmetry to obtain the recovery estimates needed to reach the central
degrees.

The code used for the central-degree bound has rate (dimension divided
by length) $1/2-g$, where $g=\Theta(n^{-1/2})$, so the recovery estimates
must remain effective as this gap tends to zero. The recovery estimate
has three stages. In
Section~\ref{sec:coordinate-recovery}, we use the area identity as stated
by Kudekar et al.~\cite{KKM} and the derivative estimate of Bourgain
and Kalai~\cite{BK} to bound the probability of failing to recover a
specified coordinate. This bound is not small enough to apply a union
bound over all $2^n$ coordinates. In Section~\ref{sec:whole-recovery},
we develop a comparison between single-coordinate and whole-word recovery
(a bit-to-block comparison) based on the work of Pfister, Sprumont and
Z\'emor~\cite{PSZ}. In Section~\ref{sec:main-proof}, we combine this
comparison with their bounds on generalized Hamming weights to show
that the whole codeword can be recovered with probability tending to
one, even at an erasure probability slightly above $1/2$. Finally, in
the same section, we use the minimum-distance estimate of Tillich and
Z\'emor~\cite{TZ} to obtain an exponential bound at $1/2$, completing
the proof of Theorem~\ref{thm:main}.

In Section~\ref{sec:consequences}, we deduce the conjecture of
Aspnes, Beigel, Furst and Rudich and discuss the remaining assertion of
Wang and Williams. In the same section, we compare our exponential
bound with a lower bound obtained from parity on a subcube.

\section{Threshold degree and erasure recovery}\label{sec:erasure-recovery}

The Reed--Muller code $\RM(n,k)$ consists of the vectors
$(P(x))_{x\in\F_2^n}$, where $P$ ranges over polynomials in $n$ variables, with coefficients
in $\F_2$, of degree at most $k$. Each such vector is a codeword: it
lists the values of $P$ at all $N=2^n$ binary inputs. Coefficients and
evaluations use arithmetic modulo 2.

Since $x_i^2=x_i$ on the binary cube, we may take $P$ to be multilinear
and write
\[
 P(x)=\sum_{\substack{S\subseteq[n]\\ |S|\leq k}}a_Sx_S,
 \qquad x_S=\prod_{i\in S}x_i,\qquad a_S\in\F_2,
\]
where $x_\varnothing=1$.

The squarefree monomials are linearly independent as functions on
$\F_2^n$, so distinct coefficient lists give distinct codewords.
Consequently, if $m_k=\sum_{j=0}^k\binom nj$, then $\RM(n,k)$ has
dimension $m_k$ and contains exactly $2^{m_k}$ codewords. The length
$N=2^n$ counts evaluation points, whereas $m_k$ counts coefficients.

For a linear code, the minimum distance is the least number of non-zero
entries in a non-zero codeword. We will also need the standard fact that
$\RM(n,k)$ has minimum distance $2^{n-k}$; see, for example,
Kudekar et al.~\cite[Section~I-B]{KKM}.

We consider the problem of recovering a codeword when some of its
coordinates have been erased. The values at the remaining coordinates
are known, and we ask whether they uniquely determine the whole codeword.
Let $V$ be the set of these visible coordinates. Let $E_V$ be the matrix
whose rows are indexed by $x\in V$, whose columns are indexed by
$S\subseteq[n]$ with $|S|\leq k$, and whose $(x,S)$ entry is $x_S$.
Recovering the codeword is equivalent to recovering its coefficient list.
Writing $a=(a_S)_{|S|\leq k}$ for the vector of coefficients of $P$,
the visible values give the linear system
\[
 E_Va=(P(x))_{x\in V}\qquad\text{over }\F_2.
\]
Recovery is possible precisely when this system has a unique solution.
Since the system is consistent, this is equivalent to
$\rank_{\F_2}E_V=m_k$. Fewer than $m_k$ visible values cannot suffice.
With more visible values, only $m_k$ independent equations are needed;
the remaining equations may be redundant.

For example, a word of $\RM(3,1)$ is the list of eight values of
$P(x)=a_\varnothing+a_1x_1+a_2x_2+a_3x_3$. The values at
$000,100,010,001$ determine all four coefficients: the first gives
$a_\varnothing$, and adding it to each of the other three values gives
$a_1,a_2,a_3$. By contrast, the four values at $000,100,010,110$ leave
$a_3$ undetermined. Thus the number of visible values alone does not
decide whether recovery is possible.

For a binary linear code $C$, erasing a coordinate means that its value
is unknown, although its position is known. Recovery is possible if the
values at the unerased coordinates determine exactly one codeword in
$C$. Let $U_p$ be the random set of erased coordinates, with each
coordinate erased independently with probability $p$.
The support $\supp w$ of a codeword $w$ is the set of coordinates
where it is non-zero. For a linear subspace $E\subseteq C$, its support
$\supp E$ is the union of the supports of its codewords.
For a set $U$ of coordinates, put
\[
 C_U=\{w\in C:\supp w\subseteq U\}.
\]
Thus $C_U$ consists of the codewords whose non-zero entries occur only
at coordinates in $U$.
If $U$ is the erased set, linearity means that adding words of $C_U$
to a given codeword produces exactly the codewords with the same
unerased entries. Recovery therefore fails precisely when $C_U$
contains a non-zero word.
We denote the probability of this failure by $B_C(p)$, when each
coordinate is erased independently with probability $p$.

The next lemma connects this binary recovery problem to real polynomial
sign representation. Its conclusion is a bound on probabilities, rather
than a procedure for recovering a real representing polynomial. The
binary code used in the lemma is $\RM(n,n-d-1)$, whose codewords are
the evaluation vectors of polynomials over $\F_2$ of degree at most
$n-d-1$.

\begin{lemma}\label{lem:regions}
For every integer $d$ with $1\leq d<n$, and a uniformly random Boolean
function $f$ on $n$ variables,
\[
 \Pr(\td(f)>d)\leq2B_{\RM(n,n-d-1)}(1/2).
\]
\end{lemma}
\begin{proof}
Since $x_i^2=x_i$ on $\{0,1\}^n$, every real polynomial can be made
multilinear without changing its values there or increasing its degree.
Let $A$ be the integer matrix with rows indexed by
$S\subseteq[n]$ with $|S|\leq d$, columns indexed by
$x\in\{0,1\}^n$, and entries $A_{S,x}=x_S=\prod_{i\in S}x_i$.
Write $A_U$ for the submatrix with columns in $U$.
In the real coefficient space, with one coordinate $a_S$ for each
monomial $x_S$, every input $x$ defines a hyperplane
$\sum_{|S|\leq d}a_Sx_S=0$. The collection of these hyperplanes is a
hyperplane arrangement. The hyperplanes are distinct. If two inputs
defined the same hyperplane, their normal vectors would be proportional.
Both have entry $1$ at the constant monomial, so they would be equal,
and their entries at $x_1,\ldots,x_n$ would show that the inputs are equal.
The arrangement is central, meaning that all its hyperplanes pass through
the origin. The coefficient vectors representing a fixed
function form an intersection of open half-spaces, which, when non-empty,
is one region of the arrangement~\cite{BV}. Thus the number $T$ of
functions represented at degree at most $d$ is the number of regions.

The Whitney and Zaslavsky formulae~\cite[Theorems~2.4 and~2.5]{Stanley}
give
\[
 T=\sum_{U\subseteq\{0,1\}^n}(-1)^{|U|-\rank_{\R}A_U},
\]
where $\rank_{\R}$ denotes rank over $\R$. Each subset of the cube occurs
as $U_{1/2}$ with probability $2^{-N}$, so dividing this sum by $2^N$ gives
\[
 \frac T{2^N}=\E\left[(-1)^{|U|-\rank_{\R}A_U}\right],
 \qquad\text{where }U=U_{1/2}.
\]
The quantity inside the expectation is $1$ when
$|U|-\rank_{\R}A_U$ is even and $-1$ when it is odd. Its expectation
is therefore one minus twice the probability of the odd case. Since $T$ counts
the functions of threshold degree at most $d$, it follows that
\[
 \Pr(\td(f)>d)=1-\frac T{2^N}
 =2\Pr\bigl(|U|-\rank_{\R}A_U\text{ is odd}\bigr).
\]
An odd value of $|U|-\rank_{\R}A_U$ is positive, so in this case the
selected columns are linearly dependent over $\R$. As noted by Baldi and
Vershynin~\cite{BV}, if these columns are dependent over $\R$, they
are also dependent over $\F_2$: a minor non-zero over $\F_2$ is an odd
integer and hence non-zero over $\R$. Hence
\[
 \Pr(\td(f)>d)
 \leq2\Pr\bigl(\text{the columns of }A_U
                         \text{ are dependent over }\F_2\bigr).
\]

The row space of $A$ over $\F_2$ is $\RM(n,d)$, because its rows
are the evaluation vectors of the monomials of degree at most $d$.
A product of two monomials of degrees at most $d$ and $n-d-1$,
respectively, omits at least one variable, so its sum over the cube is
even and hence zero over $\F_2$. The codes $\RM(n,d)$ and
$\RM(n,n-d-1)$ are therefore orthogonal over $\F_2$.
Put $C=\RM(n,n-d-1)$. Orthogonality gives $C\subseteq\ker_{\F_2}A$.
By binomial symmetry,
\[
 \dim C=\sum_{j=0}^{n-d-1}\binom nj
       =\sum_{j=d+1}^n\binom nj=N-m_d.
\]
Since $\RM(n,d)$ has dimension $m_d$, the matrix $A$ has rank $m_d$.
Its kernel therefore has dimension $N-m_d$ by rank--nullity.
Thus $C=\ker_{\F_2}A$.
For a fixed $U$, a non-zero word in $C_U$ is exactly a non-zero vector
$w$ with $Aw=0$ over $\F_2$ and $w_x=0$ outside $U$. Multiplication
by $w$ forms a linear combination of columns of $A$, with coefficients $w_x$.
These conditions therefore say precisely that the columns of $A_U$
are dependent over $\F_2$.

Taking $U=U_{1/2}$, the event that the columns of $A_U$ are dependent
over $\F_2$ is exactly the event that $C_{U_{1/2}}$ contains a non-zero
word. Its probability is $B_C(1/2)$, giving
$\Pr(\td(f)>d)\leq2B_C(1/2)$.
\end{proof}

\section{Recovery of a single coordinate}\label{sec:coordinate-recovery}

Let $C=\RM(n,k)$, whose coordinates are indexed by $\F_2^n$.
We next estimate the probability that the value at a specified erased
coordinate cannot be determined from the values at the unerased
coordinates. Suppose that the entry at coordinate $x$ is unavailable, and erase
each other entry independently with probability $t$. We cannot
recover the entry at $x$ precisely when two codewords agree at every
visible coordinate but differ at $x$. Since $C$ is linear, their
difference is a codeword $w$ with $w_x=1$, whose other non-zero entries
all lie at erased coordinates. Conversely, adding such a word to any
codeword changes its entry at $x$ without changing any visible entry.
The probability of failure to recover $x$ is therefore
\[
 h_x(t)=\Pr\bigl(\exists w\in C:\ w_x=1,\
                          \supp w\subseteq U_t\cup\{x\}\bigr).
\]

In the terminology of Kudekar et al.~\cite[Proposition~4(a)]{KKM},
$h_x$ is the EXIT function associated with coordinate $x$.
For Reed--Muller codes, the functions $h_x$ are identical for all
coordinates $x$~\cite[Proposition~7 and Lemma~23]{KKM}.
Write $h$ for this common function. In this definition the coordinate being recovered
is already unavailable, and $t$ is the erasure probability for the other
coordinates. Thus $h(t)$ is the probability of recovery failure given
that the specified coordinate is erased.

A coordinate $x$ belongs to $\supp C_{U_t}$ precisely when it is erased
and its value cannot be recovered from the unerased coordinates. This
means that there is a codeword $w$ with $w_x=1$ whose entire support
lies in $U_t$. The probability that $x$ is erased is $t$. Given this
erasure, the probability of recovery failure is $h(t)$, since the other
coordinates are still erased independently with probability $t$. Hence
\begin{equation}\label{eq:coordinate-support}
 \Pr\bigl(x\in\supp C_{U_t}\bigr)=t h(t).
\end{equation}

We use the following derivative estimate from the argument in the
proof of Theorem~27 of Kudekar et al.~\cite{KKM}.

\Needspace{9\baselineskip}
\begin{lemma}\label{lem:derivative}
For $C=\RM(n,k)$, where $0\leq k<n$, let $h$ be the common
coordinate-failure probability defined above. There is an absolute
constant $c_1>0$ such that, for all sufficiently large $n$, uniformly
in $k$,
\begin{equation}\label{eq:enhanced}
 h'(t)\geq c_1 n\ln n\,h(t)(1-h(t))
                  \qquad(1/4\leq t\leq3/4).
\end{equation}
\end{lemma}

The proof of Lemma~\ref{lem:derivative} is given at the end of this section.

\Needspace{11\baselineskip}
We combine this derivative bound with an area identity to estimate $h$
below its transition. The derivative bound controls the sharpness of
the transition, while the area identity supplies information about its
location.

Suppose that $\dim C/N=1/2-g$, where $0<g\leq1/6$.
Proposition~4(c) of Kudekar et al.~\cite{KKM} gives
\begin{equation}\label{eq:area}
 \int_0^1h(t)\,dt=\frac{\dim C}{N}=\frac12-g.
\end{equation}
A step function that jumps from zero to one at $a$ has integral $1-a$.
Thus, if $h$ were a step function, the area identity would place its jump
at $1/2+g$, and $q=1/2+g/2$ would lie a distance $g/2$ below the jump.
For the actual function $h$, the area identity alone does not show
$h(q)$ to be small. We use the derivative estimate to make the comparison
quantitative.

\Needspace{9\baselineskip}
\begin{lemma}\label{lem:coordinate}
For $C=\RM(n,k)$, where $0\leq k<n$, let $h$ be as above. Suppose that
$\dim C/N=1/2-g$, where $0<g\leq1/6$.
Then, for all sufficiently large $n$, uniformly in $k$,
\[
 h\left(\frac12+\frac g2\right)
 \leq\exp\left(-\frac{c_1}{5}ng\ln n\right),
\]
where $c_1$ is the absolute constant in Lemma~\ref{lem:derivative}.
\end{lemma}

\begin{proof}
Put $q=1/2+g/2$. Since $0<g\leq1/6$, we have $1/2<q\leq7/12$,
so Lemma~\ref{lem:derivative} applies throughout $[q,3/4]$.
We first bound the area of $1-h$ above $q$:
\begin{equation}\label{eq:missing-area}
 \int_q^1(1-h(t))\,dt
 =1-q-\int_q^1h(t)\,dt
 \geq1-q-\left(\frac12-g\right)=\frac g2,
\end{equation}
where the inequality uses~\eqref{eq:area} and the fact that $h$ is
non-negative.

Since $1-h$ is non-increasing, its average over $[q,3/4]$ is at least
its average over $[q,1]$. Hence
\[
 \int_q^{3/4}(1-h(t))\,dt
 \geq\frac{3/4-q}{1-q}\int_q^1(1-h(t))\,dt
 \geq\frac g5,
\]
where the second inequality uses~\eqref{eq:missing-area} and
$(3/4-q)/(1-q)=(1-2g)/(2(1-g))\geq2/5$, which holds since $g\leq1/6$.

To divide by $h(t)$ and take its logarithm, we first show that $h(t)>0$
for $0<t\leq1$. The constant polynomial $1$ gives the all-$1$ codeword
in $C$. If all $N-1$ entries other than $x$ are erased, we cannot
distinguish this codeword from the zero codeword at the already
unavailable coordinate $x$. Thus $h(t)\geq t^{N-1}>0$.
The remaining area forces substantial growth in $\ln h$: the derivative
inequality gives $(\ln h)'\geq c_1n\ln n\,(1-h)$.
Since $h$ can never exceed one, this growth is possible only if its
starting value $h(q)$ is sufficiently small. Dividing~\eqref{eq:enhanced}
by $h(t)$ and integrating from $q$ to $3/4$ therefore gives
\[
 \ln h(3/4)-\ln h(q)
 \geq c_1n\ln n\int_q^{3/4}(1-h(t))\,dt
 \geq\frac{c_1}{5}ng\ln n.
\]
Since $h(3/4)\leq1$, exponentiating proves the bound.
\end{proof}

\Needspace{4\baselineskip}
\begin{proof}[Proof of Lemma~\ref{lem:derivative}]
The derivative estimate we need appears in the proof of Theorem~27
of Kudekar et al.~\cite{KKM}. Their argument invokes Bourgain and
Kalai~\cite{BK}, whose proof assumes a lower bound on the variance
$h(t)(1-h(t))$ (equation~(0.25)). Our application requires the estimate
even when $h(t)(1-h(t))$ is very small. We therefore supply the
additional argument needed in this range, using independent copies,
and obtain constants independent of $k$.

The independent copies produce an event whose probability is bounded
away from both zero and one. We apply the derivative estimate to this
event and transfer the bound back to $h$. The symmetry bound needed
for this argument holds independently of the number of copies.

Let $\Omega=\F_2^n\setminus\{0\}$, the set of labels of the $N-1$
coordinates other than $0$, and let $\mathcal E$ be the family of erased
sets that prevent recovery at $0$. Write $\mu_t$ for the measure obtained
by including each coordinate independently with probability $t$, so that
$h(t)=\mu_t(\mathcal E)$. The family $\mathcal E$ is monotone, and hence
$h'(t)\geq0$. By Lemma~26 of Kudekar et al.~\cite{KKM} and its proof,
$\mathcal E$ is invariant under the natural action $x\mapsto Lx$ of
$G=\mathrm{GL}(n,\F_2)$ on $\Omega$. This action is transitive.

For a permutation group $\Gamma$ acting on a finite coordinate set $\Lambda$
and a subset $T\subseteq\Lambda$, write
$\operatorname{Orb}_\Gamma(T)=\{g(T):g\in\Gamma\}$ for the orbit of $T$ under $\Gamma$,
where $g(T)=\{g(x):x\in T\}$. Thus $|\operatorname{Orb}_\Gamma(T)|$ counts
the distinct sets obtained by applying elements of $\Gamma$ to $T$. Put
\[
 \varphi_\Gamma(s)=\min_{|T|=s}\ln|\operatorname{Orb}_\Gamma(T)|,
 \qquad
 a_1(\Gamma)=\sup\{\varphi_\Gamma(s):\varphi_\Gamma(s)>s^2\},
\]
where $s$ ranges over $1,\ldots,|\Lambda|$. These are the parameters
in equations~(0.11)--(0.12) of Bourgain and Kalai~\cite{BK}, with
$\tau=1$. For a monotone family $\mathcal A$ invariant under a transitive
group $\Gamma$, their Theorem~1, equations~(0.13)--(0.14), states that
\[
 \frac{d}{du}\mu_u(\mathcal A)
 \geq c_{\mathrm{BK}}a_1(\Gamma)\,
       \mu_u(\mathcal A)(1-\mu_u(\mathcal A)),
\]
under the condition $\ln(1/[u(1-u)])\leq c_0\ln\ln M$, where
$M=|\Lambda|$ and $c_0,c_{\mathrm{BK}}>0$ are absolute constants.
The theorem is stated without restriction on $\mu_u(\mathcal A)$,
but its proof assumes $\rho>(\ln M)^{-1}$, where
$\rho=\mu_u(\mathcal A)(1-\mu_u(\mathcal A))$. See equation~(0.25)
there. We apply the estimate only for $u\in[1/4,3/4]$ and
$\mu_u(\mathcal A)\in[1/4,3/4]$. Then $\rho\geq3/16$, so both
conditions hold for all sufficiently large $M$.

We first show that $a_1(G)\geq(n/4)\ln n$, and then establish the
same lower bound for the symmetry group of the independent copies.

Let $T\subseteq\Omega$ have $s\leq n/2$ elements, and put
$\ell=\dim\operatorname{span}(T)$.
Since an $\ell$-dimensional binary space has $2^\ell-1$ non-zero
vectors, $s\leq2^\ell-1$, and hence $\log_2(s+1)\leq\ell\leq s$.
The group $G$ takes $\operatorname{span}(T)$ to every $\ell$-dimensional
subspace. Distinct image subspaces require distinct images of $T$, so
its orbit has at least as many members as there are such subspaces.
There are at least $2^{\ell(n-\ell)}$ of these: the graphs of the
linear maps from $\F_2^\ell$ to $\F_2^{n-\ell}$ give that many distinct
subspaces. Consequently,
\[
 |\operatorname{Orb}_G(T)|
 \geq2^{\ell(n-\ell)}\geq(s+1)^{n/2}.
\]
This holds for every such $T$, and therefore
$\varphi_G(s)\geq(n/2)\ln(s+1)$. Taking $s=\lfloor\sqrt n\rfloor$,
we have, for all sufficiently large $n$,
\[
 \varphi_G(s)>\frac n4\ln n>n\geq s^2.
\]
This value qualifies in the definition of $a_1(G)$, giving
$a_1(G)\geq(n/4)\ln n$.

We now combine independent copies of the failure event to obtain an
event whose probability lies in $[1/4,3/4]$. Fix $t_0\in[1/4,3/4]$,
and put $\alpha=h(t_0)$. If $\alpha\in\{0,1\}$, the required inequality
follows from $h'(t_0)\geq0$. Suppose first that $0<\alpha\leq1/2$,
and choose $r=\lfloor1/\alpha\rfloor$. Take $r$ independent repetitions
of the recovery problem, with disjoint coordinate sets
$\Omega_1,\ldots,\Omega_r$, each a copy of $\Omega$. In every repetition,
erase each coordinate independently with probability $u$. Let
$\mathcal E_r$ be the event that recovery fails in at least one
repetition. Since the repetitions are independent, its probability is
\[
 \beta(u)=\mu_u(\mathcal E_r)=1-(1-h(u))^r.
\]
We have $1/2\leq r\alpha\leq1$. Also, concavity gives
$\ln(1-\alpha)\geq-2\alpha\ln2$ for $0\leq\alpha\leq1/2$. Thus
\[
 \frac14\leq e^{-2r\alpha\ln2}\leq(1-\alpha)^r
 \leq e^{-r\alpha}\leq e^{-1/2},
\]
and hence $\beta(t_0)\in[1/4,3/4]$. Holding $r$ fixed as $u$ varies,
direct differentiation gives
\[
 \beta'(t_0)=r(1-\alpha)^{r-1}h'(t_0)\leq r\,h'(t_0).
\]
A lower bound for $\beta'(t_0)$ will therefore give one for $h'(t_0)$.

To apply Bourgain--Kalai to $\mathcal E_r$, we need the corresponding
symmetry bound. Let $H$ be the group generated by independent actions
of $G$ within the copies and by permutations of the copies. Both
operations preserve $\mathcal E_r$, which is monotone, and $H$ acts
transitively on the enlarged coordinate set. The enlarged coordinate
set has $r(N-1)\geq N-1$ elements, so the size requirement above holds for all
sufficiently large $n$, independently of $r$.

Fix any $s$-element set $T\subseteq\Omega_1\cup\cdots\cup\Omega_r$,
where $s\leq n/2$, and put $T_i=T\cap\Omega_i$ and $s_i=|T_i|$, so that
$\sum_{i=1}^r s_i=s$. Within copy $i$, the set $T_i$ has at least
$(1+s_i)^{n/2}$ distinct images under $G$, with an empty piece
contributing $1$. These images can be chosen independently in the
different copies, and distinct combinations give distinct images of
$T$, because the copies are disjoint. Therefore
\[
 |\operatorname{Orb}_H(T)|
 \geq\prod_{i=1}^r(1+s_i)^{n/2}
 \geq\left(1+\sum_{i=1}^r s_i\right)^{n/2}
 =(1+s)^{n/2}.
\]
Thus the same orbit bound holds for $H$, independently of $r$, and
the same choice $s=\lfloor\sqrt n\rfloor$ gives
$a_1(H)\geq(n/4)\ln n$.

Applying the Bourgain--Kalai estimate to $\mathcal E_r$, and using
$\beta(t_0)(1-\beta(t_0))\geq3/16$, we obtain
\[
 r\,h'(t_0)\geq\beta'(t_0)
 \geq c_{\mathrm{BK}}a_1(H)\beta(t_0)(1-\beta(t_0))
 \geq c_1\,n\ln n
\]
where $c_1=3c_{\mathrm{BK}}/64$. Since $1/r\geq\alpha$, this implies
\[
 h'(t_0)\geq c_1\,n\ln n\,\alpha
 \geq c_1\,n\ln n\,\alpha(1-\alpha).
\]

If $1/2<\alpha<1$, instead take $r=\lfloor1/(1-\alpha)\rfloor$ and
use the event that recovery fails in every repetition. Its probability
is $\gamma(u)=h(u)^r$. The same probability bounds, with $1-\alpha$ in
place of $\alpha$, give $\gamma(t_0)\in[1/4,3/4]$. This event is also
monotone and invariant under $H$, while
\[
 \gamma'(t_0)=r\alpha^{r-1}h'(t_0)\leq r\,h'(t_0).
\]
The same application of Bourgain--Kalai now yields
\[
 h'(t_0)\geq\frac{c_1}{r}\,n\ln n
 \geq c_1\,n\ln n\,(1-\alpha)
 \geq c_1\,n\ln n\,\alpha(1-\alpha).
\]
Since $t_0$ was arbitrary, this proves~\eqref{eq:enhanced} throughout
$[1/4,3/4]$. The constant and the required lower bound
on $n$ are independent of $k$.
\end{proof}

\section{A comparison for whole-word recovery}\label{sec:whole-recovery}

To pass from recovery of a single coordinate to recovery of the whole
codeword, we use bounds on the supports of subspaces of the code.
For a binary linear code $C$ and $1\leq j\leq\dim C$, write
\[
 d_j(C)=\min\{|\supp E|:E\subseteq C,\ \dim E=j\}
\]
for the $j$th generalized Hamming weight, with the minimum taken over
linear subspaces $E$ of $C$.
Thus $d_j(C)$ is the smallest number of positions that can support $j$
independent codewords. For an erased set $U$, the dimension of $C_U$
counts the number of independent binary choices still left unresolved
by the visible entries; there are $2^{\dim C_U}$ compatible codewords.
The next lemma controls the probability that any ambiguity remains,
using the probability of a large unresolved subspace at a higher
erasure probability.

\begin{lemma}\label{lem:comparison}
For a binary linear code $C$, if $0<p<q<1$ and $K$ is an integer with
$2\leq K\leq\dim C$, then
\[
 B_C(p)\leq\Pr(\dim C_{U_q}\geq K)
             +\frac1{q-p}\sum_{j=1}^{K-1}\frac1{d_j(C)}.
\]
\end{lemma}
\begin{proof}
Put $f_j(t)=\Pr(\dim C_{U_t}\geq j)$, so that $f_1(t)=B_C(t)$.
Pfister, Sprumont and Z\'emor establish the following inequality
in the proof of their Theorem~6~\cite[equation~(10)]{PSZ}:
\[
 \int_0^1(f_j(t)-f_{j+1}(t))\,dt\leq\frac1{d_j(C)}.
\]
Summing over $j=1,\ldots,K-1$, the intermediate terms cancel, giving
\[
 \int_0^1(f_1(t)-f_K(t))\,dt
 \leq\sum_{j=1}^{K-1}\frac1{d_j(C)}.
\]
The integrand is non-negative, so this also bounds its integral over
$[p,q]$. Moreover, the functions $f_j$ are increasing, and hence
throughout this interval
\[
 f_1(t)-f_K(t)\geq f_1(p)-f_K(q).
\]
It follows that
\[
 (q-p)(f_1(p)-f_K(q))\leq\sum_{j=1}^{K-1}\frac1{d_j(C)}.
\]
Rearranging and substituting $f_1(p)=B_C(p)$ and
$f_K(q)=\Pr(\dim C_{U_q}\geq K)$ proves the result.
\end{proof}

\section{Proof of the main theorem}\label{sec:main-proof}

\noindent\textbf{The counting lower bound.}
Recall that $N=2^n$ and $m_d=\sum_{j=0}^d\binom nj$ for
$0\leq d\leq n$. By the central arrangement bound~\cite{Anthony},
the number of functions of threshold degree at most $d$ is at most
\[
 2\sum_{i=0}^{m_d-1}\binom{N-1}{i}.
\]
Dividing by $2^N$, we obtain
\begin{equation}\label{eq:central-count}
 \Pr(\td(f)\leq d)
 \leq\Pr\bigl(\operatorname{Bin}(N-1,1/2)\leq m_d-1\bigr).
\end{equation}
Now put $d_0=\lfloor n/2\rfloor-1$ and $m=m_{d_0}$.
The central binomial estimates give
$m/N\leq1/2-c_2/\sqrt n$ for a fixed positive absolute constant $c_2$
and all sufficiently large $n$. Hoeffding's inequality therefore gives
\[
 \Pr(\td(f)\leq d_0)
 \leq\Pr\bigl(\operatorname{Bin}(N-1,1/2)\leq m-1\bigr)
 \leq\exp(-2c_2^2N/n).
\]

\medskip
\noindent\textbf{Single-coordinate recovery.}
For the upper bound, fix $d=\lceil n/2\rceil$ and
$C=\RM(n,n-d-1)$. Since $n-d-1=d_0$, this code has dimension $m$.
By Lemma~\ref{lem:regions}, it suffices to bound $B_C(1/2)$.
Define $g$ by
\[
 \frac{\dim C}N=\frac mN=\frac12-g.
\]
Therefore
\[
 g=\begin{cases}
  2^{-n-1}\binom n{n/2},&n\text{ even},\\[3pt]
  2^{-n}\binom n{(n-1)/2},&n\text{ odd}.
 \end{cases}
\]
The earlier bound on $m/N$ gives $g\geq c_2/\sqrt n$ for large $n$.
The central binomial estimates also give $g=O(n^{-1/2})$.
The code $C$ has order $k=d_0=\lfloor n/2\rfloor-1$.
As in Section~\ref{sec:erasure-recovery}, each codeword is specified by
$m=N(1/2-g)$ binary coefficients, and the visible values give a linear
system in these coefficients.
At erasure probability $1/2$, the expected number of equations is
$N/2$, exceeding the number of unknowns by $Ng$. This surplus suggests
that recovery may remain possible at erasure probabilities slightly
above $1/2$, and motivates our choices of $p$ and $q$ below. The
equation count alone does not establish recovery. As the example in
Section~\ref{sec:erasure-recovery} shows, we need the visible equations
to have rank $m$.

At erasure probability $1/2+g$, the expected number of visible values
equals $m$. This is also the ideal transition location suggested by
the area identity. We will work below this value, leaving a gap of $g/2$
for the single-coordinate estimate at $q$ and a further gap of $g/4$
to pass from recovery of one coordinate to recovery of the whole word at $p$.
We will first show that recovery fails with probability tending to
zero even when the erasure probability is slightly greater than $1/2$.
A minimum-distance estimate will then give the required exponential
bound at $1/2$.

We choose
\[
 p=\frac12+\frac g4,\qquad q=\frac12+\frac g2.
\]
We will bound $B_C(p)$ by applying Lemma~\ref{lem:comparison}
with these values of $p$ and $q$.
Let $h$ be as in Lemma~\ref{lem:coordinate}.
Since $g=O(n^{-1/2})$, we have $g\leq1/6$ for large $n$, so the lemma
applies. Using $g\geq c_2/\sqrt n$ in its bound gives, with
$c_3=c_1c_2/5$,
\begin{equation}\label{eq:bit}
 h(q)\leq\exp(-c_3\sqrt n\ln n).
\end{equation}

\medskip
\noindent\textbf{Whole-word recovery.}
We have bounded the probability of failing to recover a specified
coordinate. We now need to bound $B_C(p)$, the probability of failing
to recover the whole codeword. A union bound at $q$ would give
$B_C(q)\leq Nqh(q)$, but substituting~\eqref{eq:bit} into this
bound does not give a quantity tending to zero: $N=2^n$ grows too
quickly. We instead control the probability that a large subspace
of possibilities remains, and use generalized weights to handle
the smaller dimensions after decreasing the erasure probability.
We choose $K$ so that both terms on the right-hand side of
Lemma~\ref{lem:comparison} tend to zero.
Set
\[
 K=\left\lfloor2^{\,n-6\sqrt{n\log_2 n}}\right\rfloor.
\]
For all sufficiently large $n$, $2\leq K\leq\dim C$.

We first bound $\Pr(\dim C_{U_q}\geq K)$ using our estimate for
$h(q)$. Summing~\eqref{eq:coordinate-support} over the $N$ coordinates gives
\[
 \E|\supp C_{U_q}|=Nqh(q).
\]
A subspace supported on $s$ coordinates has dimension at most $s$.
Thus, whenever $\dim C_{U_q}\geq K$, the support contains at least
$K$ coordinates. Its expected size is therefore at least $K$ times
the probability of this event, giving
\[
 \Pr(\dim C_{U_q}\geq K)
 \leq\frac{\E|\supp C_{U_q}|}{K}
 =\frac{Nqh(q)}K
 \leq\frac{Nh(q)}K.
\]
Using~\eqref{eq:bit} and the definition of $K$, we obtain
\[
 \frac{Nh(q)}K
 \leq2\exp\left(-c_3\sqrt n\ln n
                  +6\sqrt{n\ln n\ln2}\right)=o(1),
\]
since the negative term in the exponent grows faster in magnitude
than the positive term. The factor $\ln n$ from Bourgain and Kalai
is used here: without it, the negative term would have order $\sqrt n$
and could not dominate the positive term of order $\sqrt{n\ln n}$.
The range of the generalized-weight estimate used below prevents us
from taking a substantially larger $K$.

We now bound the second term in Lemma~\ref{lem:comparison}.
Taking $u=6\sqrt{\log_2 n/n}$ in Corollary~10 of Pfister, Sprumont
and Z\'emor~\cite{PSZ} gives
\[
 d_j(C)\geq j2^{u^2n/10}=j n^{18/5}\qquad(1\leq j<K).
\]
The conditions of their bound hold: the order $n-d-1$ is at most
$n/2$, $6\sqrt{\log_2 n/n}\leq u\leq1/2$ for large $n$, and
$j<K\leq2^{n-un}$.
Since $\sum_{j<K}1/j\leq1+\ln K\leq1+n\ln2$, while
$q-p=g/4$ and $g\geq c_2/\sqrt n$, it follows that
\[
 \frac1{q-p}\sum_{j=1}^{K-1}\frac1{d_j(C)}
 \leq\frac{4(1+n\ln2)}{g n^{18/5}}
 =O(n^{3/2-18/5})=O(n^{-21/10})=o(1).
\]
Both terms in Lemma~\ref{lem:comparison} tend to zero, so
$B_C(p)=o(1)$.

\medskip
\noindent\textbf{The exponential bound at $1/2$.}
Since $B_C$ is continuous, with $B_C(0)=0$ and $B_C(1)=1$, choose
$b\in(0,1)$ such that $B_C(b)=1/2$. For large $n$, $B_C(p)<1/2$,
so monotonicity gives $b>p$, and hence $b-1/2\geq g/4$.
By the minimum-distance formula recalled in
Section~\ref{sec:erasure-recovery}, the minimum distance of $C$ is
\[
 D=2^{d+1}=2^{\lceil n/2\rceil+1}.
\]
The estimate of Tillich and Z\'emor controls the decay of $B_C(t)$ for
$t<b$, where $B_C(b)=1/2$. Since $b-1/2\geq g/4$, it gives an exponent
of order $Dg^2$, explaining the scale $2^{n/2}/n$ in the theorem.
Write $\Phi$ for the standard normal distribution function.
Since $b>1/2$, the estimate of Tillich and Z\'emor~\cite[Theorem~5.2,
equation~(5.1)]{TZ} gives
\[
 B_C(1/2)\leq
 \Phi\!\left(\sqrt{2D}
       \left(\sqrt{-\ln b}-\sqrt{\ln2}\right)\right).
\]
Since $x^2\ln(1/x)\leq1/(2e)$ for $0<x<1$, we have
\[
 \sqrt{\ln2}-\sqrt{-\ln b}
 =\int_{1/2}^{b}\frac{dx}{2x\sqrt{-\ln x}}
 \geq\sqrt{\frac e2}\,(b-1/2).
\]
The standard Gaussian tail bound $\Phi(-z)\leq\tfrac12e^{-z^2/2}$
for $z\geq0$ therefore gives
\[
 B_C(1/2)\leq\tfrac12\exp\left(-\frac{eD}{2}(b-1/2)^2\right).
\]
Since $b-1/2\geq g/4$, Lemma~\ref{lem:regions} and the lower bound
on $g$ therefore give
\[
 \Pr\bigl(\td(f)>\lceil n/2\rceil\bigr)
 \leq\exp(-eDg^2/32)
 \leq\exp\left(-\frac{e c_2^2}{16}\frac{2^{n/2}}n\right).
\]
Choose the constant in the theorem as
$c=e c_2^2/32$. For sufficiently large $n$, the sum of
this upper-tail bound and the lower-tail bound $\exp(-2c_2^2N/n)$ is
at most $\exp(-c2^{n/2}/n)$, proving the theorem.\nobreak\hfill\qedsymbol

\section{Consequences and the exceptional probability}\label{sec:consequences}

\noindent\textbf{The Wang--Williams conjecture.}
Theorem~\ref{thm:main} proves the conjecture of Wang and
Williams~\cite{WW} for even $n$ and its concentration assertion for odd
$n$: almost every Boolean function has threshold degree among the
central one or two integers. Their additional assertion for
odd $n$, that the proportions at the two central degrees both tend to
$1/2$, remains open. The counting argument already supplies one side
of the bound needed for this balance.

For odd $n=2\nu+1$, the space of polynomials of degree at most $\nu$ on
the cube has dimension $\sum_{j=0}^{\nu}\binom{n}{j}=N/2$.
By symmetry of $\operatorname{Bin}(N-1,1/2)$, the counting
bound~\eqref{eq:central-count} therefore gives
\[
 \Pr\bigl(\td(f)\leq \nu\bigr)\leq\frac12.
\]
Combining this with Theorem~\ref{thm:main}, we obtain
\[
 \Pr\bigl(\td(f)=\nu\bigr)\leq\frac12,
 \qquad
 \Pr\bigl(\td(f)=\nu+1\bigr)\geq\frac12-o(1),
\]
as $n$ tends to infinity through odd values. To complete the
Wang--Williams conjecture, it remains to prove the reverse asymptotic
bound
\[
 \Pr\bigl(\td(f)=\nu\bigr)\geq\frac12-o(1).
\]
This would force both central proportions to tend to $1/2$.

\medskip
\Needspace{7\baselineskip}
\noindent\textbf{The Aspnes--Beigel--Furst--Rudich conjecture.}
Theorem~\ref{thm:main} also gives the same concentration bound for weak
threshold degree. Write $\td_{\mathrm w}(f)$ for the weak threshold
degree defined in Section~\ref{sec:introduction}. Put
$\chi(x)=(-1)^{x_1+\cdots+x_n}$ for full parity. The change of variables
$y_i=1-2x_i$ between $\{0,1\}^n$ and $\{-1,1\}^n$ preserves degree
bounds in both directions and identifies $\chi(x)$ with $\prod_{i=1}^n y_i$.
Thus the duality of Aspnes, Beigel, Furst and
Rudich~\cite[Lemma~2.5]{ABFR} applies to our convention and gives
\[
 \td_{\mathrm w}(f)+\td(f\chi)=n.
\]
Since $f\chi$ is uniformly distributed whenever $f$ is, and reflection
about $n/2$ preserves the central interval, the claimed concentration
follows from Theorem~\ref{thm:main}. For even $n$, applying the
concentration bounds to both degrees shows that almost every Boolean
function satisfies
\[
 \td(f)=\td_{\mathrm w}(f)=n/2.
\]
This proves the conjecture of Aspnes, Beigel, Furst and
Rudich~\cite{ABFR} stated in Section~\ref{sec:introduction}.

\medskip
\Needspace{7\baselineskip}
\noindent\textbf{A parity lower bound.}
For $n\geq2$, fix a coordinate subcube of dimension
$\ell=\lceil n/2\rceil+1$, and put $D=2^\ell$ as in
Section~\ref{sec:main-proof}.
Parity on this subcube has threshold degree
$\ell$: it is orthogonal to every polynomial of smaller degree, so no
such polynomial can agree with it in sign. Restriction to a coordinate
subcube preserves the degree bound. Thus, if the restriction of $f$
to this subcube is parity or its negative, then $\td(f)\geq\ell$.
The subcube has $D=2^\ell$ points, and the restriction of a uniformly
random $f$ is uniform among the $2^D$ Boolean functions on it.
Exactly two of these functions are parity and its negative, so this
event has probability $2/2^D=2^{1-D}=2^{1-2^\ell}$. Hence
\[
 \Pr\bigl(\td(f)>\lceil n/2\rceil\bigr)\geq2^{1-D}.
\]
The exponent in this lower bound is of order $D$, whereas that in our
upper bound is of order $D/n$. The two exponents therefore differ by
a factor of order $n$.

\section*{Funding statement}
This work received no specific grant from any funding agency,
commercial or not-for-profit sectors.

\section*{Competing interests}
The author declares none.

\section*{Data availability}
No data were generated or analysed in this work.

\section*{Declaration of generative AI use}
Anthropic's Claude and OpenAI's ChatGPT were used to assist with literature
research, improvements to the clarity of exposition, and the development
and assessment of some ideas and arguments. In particular, they helped
resolve a concern about applying the Bourgain--Kalai influence inequality
at small variance by developing the independent-copy amplification
argument given in the proof of Lemma~\ref{lem:derivative}. The author checked
the arguments and takes responsibility for the content of the paper.

\end{document}